\documentclass[runningheads,a4paper]{llncs}

\usepackage{amssymb}
\usepackage{graphicx}
\usepackage{helvet}
\usepackage{courier}
\usepackage{amssymb}
\usepackage{amsmath,dsfont}
\usepackage{graphicx}
\usepackage{url,color}
\usepackage{multirow}
\usepackage{adjustbox}

\urldef{\mailsa}\path| {Sergey.Polyakovskiy, Frank.Neumann}@adelaide.edu.au|    
\newcommand{\keywords}[1]{\par\addvspace\baselineskip
\noindent\keywordname\enspace\ignorespaces#1}

\DeclareMathOperator*{\mycup}{\cup}
\newcommand{\ignore}[1]{}

\begin{document}

\newcommand{\NKPc}{$\text{NKP}^c$}
\newcommand{\NKPu}{$\text{NKP}^u$}
\newcommand{\ANKP}{$\text{NKP}_{\tau}^{a}$}
\newcommand{\ANKPh}{$\text{NKP}_{100}^{a}$}
\newcommand{\ANKPt}{$\text{NKP}_{1000}^{a}$}
\newcommand{\ENKP}{$\text{NKP}^{e}$}

\mainmatter

\title{Packing While Traveling: \\ Mixed Integer Programming for a Class of Nonlinear Knapsack Problems\thanks{As to appear in the proceedings of the 12th International Conference on Integration of Artificial Intelligence (AI) and Operations Research (OR) techniques in Constraint Programming (CPAIOR 2015), LNCS 9075, pp. 330–344, 2015. DOI: 10.1007/978-3-319-18008-3\textunderscore23}}
\titlerunning{Packing While Traveling}
\author{Sergey Polyakovskiy\and Frank Neumann}
\authorrunning{Sergey Polyakovskiy\and Frank Neumann}
\institute{Optimisation and Logistics\\ School of Computer Science, \\ The University of Adelaide, Australia\\ 
\mailsa\\ \url{http://cs.adelaide.edu.au/~optlog/}}

\toctitle{Packing While Traveling}
\tocauthor{Sergey Polyakovskiy}
\maketitle

\begin{abstract}
Packing and vehicle routing problems play an important role in the area of supply chain management. In this paper, we introduce a non-linear knapsack problem that occurs when packing items along a fixed route and taking into account travel time. We investigate constrained and unconstrained versions of the problem and show that both are $\mathcal{NP}$-hard. In order to solve the problems, we provide a pre-processing scheme as well as exact and approximate mixed integer programming (MIP) solutions. Our experimental results show the effectiveness of the MIP solutions and in particular point out that the approximate MIP approach often leads to near optimal results within far less computation time than the exact approach.
\end{abstract}

\section{Introduction}
Knapsack problems belong to the core combinatorial optimization problems and have been frequently studied in the literature from the theoretical as well as experimental perspective~\cite{GareyJ79,Martello90}. While the classical knapsack problem asks for the maximizing of a linear pseudo-Boolean function under one linear constraint, different generalizations and variations have been investigated such as the multiple knapsack problem~\cite{ChekuriK05} and multi-objective knapsack problems~\cite{ErlebachKP01}.

Furthermore, knapsack problems with nonlinear objective functions have been studied in the literature from different perspectives~\cite{BretthauerS02}. Hochbaum~\cite{Hochbaum95} considered the problem of maximizing a separable concave objective function subject to a packing constraint and provided an FPTAS. An exact approach for a nonlinear knapsack problem with a nonlinear term penalizing the excessive use of the knapsack capacity has been given in~\cite{Elhedhli05}.

Nonlinear knapsack problems also play a key-role in various Vehicle Routing Problems (VRP). In recent years, the research on dependence of the fuel consumption on different factors, like a travel velocity, a load's weight and vehicle's technical specifications, in various VRP has gained attention from the operations research community. Mainly, this interest is motivated by a wish to be more accurate with the evaluation of transportation costs, and therefore to stay closer to reality. Indeed, an advanced precision would immediately benefit to transportation efficiency measured by the classic petroleum-based costs and the novel greenhouse gas emission costs. In VRP in general, and in the Green Vehicle Routing Problems (GVRP) that consider energy consumption in particular, given are a depot and a set of customers which are to be served by a set of vehicles collecting (or delivering) required items. While the set of items is fixed, the goal is to find a route for each vehicle such that the total size of assigned items does not exceed the vehicle's capacity and the total traveling cost over all vehicles is minimized. See \cite{Lin14} for an extended overview on VRP and GVRP. 
Oppositely, we address the situation with one vehicle whose route is fixed but the items can be either collected or skipped. Specifically, this situation represents a class of nonlinear knapsack problems and considers trade-off between the profits of collected items and the traveling cost affected by their total weight.
The non-linear packing problem arises in some practical applications. For example, a supplier having a single truck has to decide on goods to purchase going through the constant route in order to maximize profitability of later sales.  

Our precise setting is inspired by the recently introduced Traveling Thief Problem (TTP)~\cite
{BonyadiMB13} which combines the classical Traveling Salesperson Problem (TSP) with the 0-1 Knapsack Problem (KP). 
The TTP involves searching for a permutation of the cities and a packing such that the resulting profit is maximal.
The TTP has some relation to the Prize Collecting TSP~\cite{Balas89} where a decision is made on whether to visit a given city. In the Prize Collecting TSP, a city-dependent reward is obtained when a city is visited and a city-dependent penalty has to be paid for each non-visited city. 
In contrast to this, the TTP requires that each given city is visited. Furthermore, each city has a set of available items with weights and profits and a decision has to be made which items to pick. A selected item contributes its profit to the overall profit. However, the weight of an item leads to a higher transportation cost, and therefore has a negative impact on the overall profit.

Our non-linear knapsack problem uses the same cost function as the TTP, but assumes a fixed route. It deals with the problem which items to select when giving a fixed route from an origin to a destination.  Therefore, our approach can also be applied to solve the TTP by using the non-linear packing approach as a subroutine to solve the packing part. 
Our experimental investigations are carried out on the benchmark set for the traveling thief problem~\cite{Polyakovskiy14} where we assume that the route is fixed.

The paper is organized as follows. In Section~\ref{sec:prob}, we introduce the nonlinear knapsack problems and show in Section~\ref{sec:COMPL}  that they are $\mathcal{NP}$-hard. In Section~\ref{sec:COMPL}, we provide a pre-processing scheme which allows to identify unprofitable and compulsory items. Sections~\ref{sec:ES} and \ref{sec:AS} introduce our mixed-integer program based approaches to solve the problem exactly and approximately. We report on the results of our experimental investigations in Section~\ref{sec:CE} and finish with some conclusions.

\section{Problem Statement} \label{sec:prob}

We consider the following non-linear packing problem inspired by the traveling thief problem~\cite{BonyadiMB13}. Given is a route $N=\left(1, 2,\ldots, n+1 \right)$ as a sequence of $n+1$ cities where all cities are unique and distances $d_i>0$ between pairs of consecutive cities $(i, i+1)$, $1\leq i \leq n$. 
There is a vehicle which travels through the cities of $N$ in the order of this sequence starting its trip in the first city and ending it in the city $n+1$ as a destination point. Every city $i$, $1 \leq i \leq n$, contains a set of distinct items $M_i=\{e_{i1},\ldots,e_{im_i}\}$ and we denote by \mbox{$M = \displaystyle \mycup_{1 \leq i \leq n} M_i$} set of all items available at all cities. Each item $e_{ik} \in M$ has a positive integer profit $p_{ik}$ and a weight $w_{ik}$. The vehicle may collect a set of items on the route such that the total weight of collected items does not exceed its capacity $W$. Collecting an item $e_{ik}$ leads to a profit contribution $p_{ik}$, but increases the transportation cost as the weight $w_{ik}$ slows down the vehicle. The vehicle travels along $(i, i+1)$, $1 \leq i \leq n$, with velocity $v_i \in [\upsilon_{\min}, \upsilon_{\max}]$ which depends on the weight of the items collected in the first $i$ cities.
The goal is to find a subset of $M$ such that the difference between the profit of the selected items and the transportation cost is maximized. 

To make the problem precise we give a nonlinear binary integer program formulation. The program consists of one variable $x_{ik}$ for each item $e_{ik} \in M$ where $e_{ik}$ is chosen iff $x_{ik}=1$. A decision vector $X=\left(x_{11},\ldots,x_{nm_n}\right)$ defines the packing plan as a solution. If no item has been selected, the vehicle travels with its maximal velocity $\upsilon_{max}$. Reaching its capacity $W$, it travels with minimal velocity $\upsilon_{min}>0$.  The velocity depends on the weight of the chosen items in a linear way. The travel time $t_i = \frac{d_i}{v_i}$ along $(i, i+1)$ is the ratio of the distance $d_i$ and the current velocity 
$$\upsilon_i=\upsilon_{max}-\nu \sum_{j=1}^i \sum_{k=1}^{m_j} w_{jk} x_{jk}$$ 
which is determined by the weight of the items collected in cities $1, \ldots, i$. Here, $\nu = \frac{\upsilon_{max}-\upsilon_{min}}{W}$ is a constant value defined by the input parameters. 
The overall transportation cost is given by the sum of the travel costs along $(i, i+1)$, $1 \leq i \leq n$, multiplied by a given rent rate $R>0$.
In summary, the problem is given by  the following nonlinear binary program ({\NKPc}). 

\vspace{-.25cm}
{\footnotesize
\begin{flalign}
\mbox{max} &  \displaystyle\sum_{i=1}^n \left(\displaystyle\sum_{k=1}^ {m_i} p_{ik} x_{ik}
- \frac{Rd_i}{\upsilon_{max}-\nu \displaystyle\sum_{j=1}^i \displaystyle\sum_{k=1}^{m_j} w_{jk} x_{jk}}\right) \label{eq:1}
\\ 
\mbox{s.t.} & \displaystyle\sum_{i=1}^n \displaystyle\sum_{k=1}^{m_i} w_{ik} x_{ik} \leq W \label{eq:2}
\\ \nonumber
& x_{ik} \in \left\{0,1\right\}, \; e_{ik} \in M
\end{flalign}}
\vspace{-.25cm}

We also consider the unconstrained version {\NKPu} of {\NKPc} where we set $W \geq \sum_{e_{ik} \in M} w_{ik}$ such that every selection of items yields a feasible solution. Given a real value $B$, the decision variant of {\NKPc} and {\NKPu} has to answer the question whether the value of (\ref{eq:1}) is at least $B$.

\section{Complexity of the Problem} \label{sec:COMPL}

In this section, we investigate the complexity of {\NKPc} and {\NKPu}. {\NKPc} is NP-hard as it is a generalization of the classical NP-hard 0-1 knapsack problem ~\cite{Martello90}. In fact, assigning zero either to the rate $R$ or to every distance value $d_i$ in {\NKPc}, we obtain KP. Our contribution is the proof that the unconstrained version {\NKPu} of the problem remains $\mathcal{NP}$-hard. We show this by reducing the $\mathcal{NP}$-complete \textit{subset sum problem} (SSP) to the decision variant of {\NKPu} which asks whether there is a solution with objective value at least $B$. The input for SSP is given by $q$ positive integers $S=\left\{s_1, \ldots, s_q\right\}$ and a positive integer $Q$. The question is whether there exists a vector $X=\left(x_1, \ldots, x_q\right)$, $x_{k} \in \left\{0,1\right\}$, $1\leq k\leq q$, such that $\sum_{k=1}^q {s_kx_k} = Q$.

\begin{theorem}
{\NKPu} is $\mathcal{NP}$-hard.
\end{theorem}

\begin{proof} 
We reduce SSP to the decision variant of {\NKPu} which asks whether there is a solution of objective value at least $B$. 

We encode the instance of SSP given by the set of integers $S$ and the integer $Q$ as the instance $I$ of {\NKPu} having two cities. The first city contains $q$ items while the second city is a destination point free of items. We set the distance between two cities $d_1=1$, and set $p_{1k}=w_{1k}=s_k$, $1 \leq k \leq q$ and $W = \sum_{k=1}^q s_k$. Subsequently, we set $\upsilon_{max}=2$ and $\upsilon_{min}=1$ which implies $\nu = 1/W$ and define $R^*=W\left(2-Q/W\right)^2$.

Consider the nonlinear function $f_{R^*} \colon \left[0,W\right] \rightarrow \mathds{R}$ defined as

\vspace{-.25cm}
{\footnotesize
\begin{align} \label{func:1}
f_{R^*}\left(w\right)=w-\frac{R^*}{2- w/W}.
\end{align}}
\vspace{-.25cm}

\noindent $f_{R^*}$ defined on the interval $\left[0,W\right]$ is a continuous convex function that reaches its unique maximum in the point $w^*=W \cdot (2-\sqrt{ R^*/W})$ = Q, i.e. \mbox{$f_{R^*}\left(w \right)<f_{R^*}\left(w^*\right)$} for $w \in [0,W]$ and $w \not = w^*$. Then $f_{R^*}(Q)$ is the maximum value for $f_{R^*}$ when being restricted to integer input, too. Therefore, we set $B=f_{R^*}(Q)$ and the objective function for {\NKPu} is given by

\vspace{-.25cm}
{\footnotesize
\begin{align} \label{func:2}
g_{R^*}\left(x\right)=\displaystyle\sum_{k=1}^q {p_kx_k}-\frac{R^*}{2- \frac{1}{W} \displaystyle\sum_{k=1}^q {w_kx_k}}.
\end{align}}
\vspace{-.25cm}

There exists an $x \in \{0,1\}^q$ such that $g_{R^*}(x) \geq B= f_{R^*}(Q)$ iff $\sum_{k=1}^q s_kx_k=\sum_{k=1}^q w_{1k}x_k=\sum_{k=1}^q p_{1k}x_k=Q$. Therefore, the instance of SSP has answer YES iff the optimal solution of the {\NKPu} instance $I$ has objective value at least $B=f_{R^*}\left(Q\right)$. Obviously, the reduction can be carried out in polynomial time which completes the proof.
\qed
\end{proof}

\section{Pre-processing} \label{sec:RS}
We now provide a pre-processing scheme to identify items of a given instance $I$ that can be either directly included or discarded. Removing such items from the optimization process can significantly speed up the algorithms. Our pre-processing will allow to decrease the number of decision variables for mixed integer programming approaches described in Sections~\ref{sec:ES} and \ref{sec:AS}. We distinguish between two kinds of items that are identified in the pre-processing: \emph{compulsory} and \emph{unprofitable} items. We call an item \emph{compulsory} if its inclusion in any packing plan increases the value of the objective function, and call an item \emph{unprofitable} if its inclusion in any packing plan does not increase the value of the objective function. Therefore, an optimal solution has to include all compulsory items while all unprofitable items can be discarded. 

In order to identify \emph{compulsory} and \emph{unprofitable} items, we consider the total travel cost that a set of items produces.
\begin{definition}[Total Travel Cost]
Let $O \subseteq M$ be a subset of items. We define the total travel cost along route $N$ when the items of $O$ are selected as
\vspace{-.25cm}
{\footnotesize
\begin{align} 
\nonumber t_O= R \cdot \sum_{i=1}^n  \frac{d_i}{\upsilon_{max}-\nu \sum_{j=1}^i \sum_{e_{jk} \in O_j} w_{jk}},
\end{align}}
\vspace{-.25cm}

\noindent where $O_j = M_j \cap O$, $1 \leq j \leq n$, is the subset of $O$ selected at city $j$.
\end{definition}

We identify compulsory items for the unconstrained case according to the following proposition.
\begin{proposition}[Compulsory Item] \label{prop2}
Let $I$ be an arbitrary instance of {\NKPu}.
If $p_{ik} > R\left(t_{M} - t_{M \setminus \left\{ e_{ik} \right\} } \right)$, then $e_{ik}$ is a compulsory item.
\end{proposition}

\begin{proof}
We work under the assumption that $p_{ik} > R\left(t_{M}-t_{M \setminus \left\{ e_{ik} \right\} }\right)$ holds. In the case of {\NKPu}, all the existing items can fit into the vehicle at once and all subsets $O \subseteq M$ are feasible. Let $M^* \subseteq M \setminus \left\{ e_{ik} \right\}$ be an arbitrary subset of items excluding $e_{ik}$, and consider $t_{M \setminus M^*}$ and $t_{M \setminus M^* \setminus \left\{ e_{ik}\right\} }$, respectively. Since the velocity depends linearly on the weight of collected items and the travel time $t_i=d_i/v_i$ along $(i,i+1)$ depends inversely proportional on the velocity $v_i$, we have $\left(t_{M}-t_{M \setminus \left\{ e_{ik} \right\} }\right) \geq \left(t_{M \setminus M^*}-t_{M \setminus M^* \setminus \left\{ e_{ik}\right\} }\right)$. This implies that $p_{ik}>R\left(t_{M \setminus M^*}-t_{M \setminus M^* \setminus \left\{ e_{ik}\right\} }\right)$ holds for any subset $M \setminus M^*$ of items which completes the proof.
\qed
\end{proof}

For the unconstrained variant {\NKPu}, Proposition~\ref{prop2} is valid to determine whether the item $e_{ik}$ is able to cover by its $p_{ik}$ the largest possible transportation costs it may generate when has been selected in $X$. Here, the largest possible transportation costs are computed via the worst case scenario when all the possible items are selected along with $e_{ik}$, and therefore when the vehicle has the maximal possible load and the least velocity.

Based on a given instance, we can identify unprofitable items for the constrained and unconstrained case according to the following proposition.

\begin{proposition}[Unprofitable Item, Case 1] \label{prop1}
Let $I$ be an arbitrary instance of {\NKPc} or {\NKPu}. If $p_{ik} \leq R\left(t_{\{e_{ik}\}} - t_\emptyset \right)$, then $e_{ik}$ is an unprofitable item.
\end{proposition}

\begin{proof} 
We assume that $p_{ik}\leq R\left(t_{\left\{e_{ik}\right\}}-t_\emptyset\right)$ holds. Let $M^* \subseteq M \setminus \left\{e_{ik} \right\}$ denote an arbitrary subset of items excluding $e_{ik} $ such that \mbox{$w_{ik}+\sum_{e_{jl}\in M^*} w_{jl} \leq W$} holds. We consider $t_{M^*\cup \left\{ e_{ik} \right\} }$ and $t_{M^*}$. Since the velocity depends linearly on the weight of collected items and the travel time   $t_i=d_i/v_i$ along $(i,i+1)$ depends inversely proportional on the velocity $v_i$, the inequality $\left(t_{\left\{e_{ik}\right\}}-t_\emptyset\right) \leq \left(t_{M^*\cup \left\{ e_{ik} \right\} }-t_{M^*}\right)$ holds. Therefore, $p_{ik}\leq R\left(t_{M^*\cup \left\{ e_{ik} \right\}}-t_{M^*}\right)$ holds for any $M^* \subseteq M \setminus \left\{e_{ik} \right\}$ which completes the proof.
\qed
\end{proof}

Proposition~\ref{prop1} helps to determine whether the profit $p_{ik}$ of the item $e_{ik}$ is large enough to cover the least transportation costs it incurs when selected in the packing plan $X$. In this case, the least transportation costs result from accepting the selection of $e_{ik}$ as only selected item in $X$ versus accepting empty $X$ as a solution.

Having all compulsory items included in the unconstrained case according to Proposition~\ref{prop2}, we can identify further unprofitable items. This is the case, as the inclusion of compulsory items already increases the travel time and therefore reducing the positive contribution to the overall objective value.

\begin{proposition}[Unprofitable Item, Case 2] \label{prop3}
Let $I$ be an arbitrary instance of {\NKPu} and $M^c$ be the set of all compulsory items. If $p_{ik} \leq R\left(t_{M^c \cup \left\{e_{ik}\right\} } - t_{M^c} \right)$, then $e_{ik}$ is an unprofitable item.
\end{proposition}

\begin{proof}
We assume that $p_{ik} \leq R\left(t_{M^c \cup \left\{e_{ik}\right\} }-t_{M^c} \right)$ holds.
Recall that in the case of {\NKPu}, all the existing items can fit into the vehicle at once and all subsets $O \subseteq M$ are feasible. 
Let $M^* \subseteq M \setminus \left\{ M^c \cup \left\{e_{ik}\right\} \right\}$ be an arbitrary subset of $M$ that does not include any item of  $M^c \cup \left\{e_{ik} \right\}$ and consider $t_{M^c \cup M^*}$ and $t_{M^c \cup M^* \cup \left\{ e_{ik} \right\} }$.
Since the velocity depends linearly on the weight of collected items and the travel time   $t_i=d_i/v_i$ along $(i,i+1)$ depends inversely proportional on the velocity $v_i$, we have
$\left( t_{M^c \cup \left\{e_{ik}\right\} }-t_{M^c} \right) \leq \left(t_{M^c \cup M^* \cup \left\{ e_{ik} \right\} }-t_{M^c \cup M^*} \right)$. Hence, we have $p_{ik} \leq R\left(t_{M^c \cup M^* \cup \left\{e_{ik} \right\}}-t_{M^c \cup M^*} \right)$ for any $M^*\subseteq M \setminus \left\{ M^c \cup \left\{e_{ik}\right\} \right\}$ which completes the proof.
\qed
\end{proof}

Proposition~\ref{prop3} determines for the {\NKPu} problem whether the profit $p_{ik}$ of the item $e_{ik}$ is large enough to cover the least transportation costs resulted from its selection along with all known compulsory items. Specifically, in Proposition~\ref{prop3} the list transportation costs follow from accepting the selection of $e_{ik}$ along with the set of compulsory items $M^c$ in $X$ versus accepting just the selection of $M^c$ as a solution. 

It is important to note that Proposition~\ref{prop1} can reduce {\NKPc} problem to {\NKPu} by excluding items such that the sum of the weights of all remaining items does not exceed the weight bound $W$. In this case, Propositions~\ref{prop2} and \ref{prop3} can be applied iteratively to the remaining set of items until no compulsory or unprofitable item is found. Before applying our approaches given in Section~\ref{sec:ES} and \ref{sec:AS}, we remove all unprofitable and compulsory items using these preprocessing steps.

\section{Exact Solution} \label{sec:ES}
Both {\NKPc} and {\NKPu} contain nonlinear terms in the objective function, and therefore are nonlinear binary programs. They belong to the specific class of fractional binary programming problems for which several efficient reformulation techniques exist to handle nonlinear terms. We follow the approach of \cite{Li94} and \cite{Tawarmalani02} to reformulate {\NKPc} and {\NKPu} as a linear mixed 0-1 program. 

The denominator of each fractional term in (\ref{eq:1}) is not equal to zero since $\upsilon_{min}>0$. We introduce the auxiliary real-valued variables $y_i$, $i=1, \ldots, n$,  such that $y_i = 1/\left(\upsilon_{max}-\nu \sum_{j=1}^i \sum_{k=1}^{m_j} w_{jk} x_{jk} \right)$. The variables $y_i$ express the travel time per distance unit along $\left(i, i+1\right)$. According to \cite{Li94}, we can reformulate {\NKPc} as a mixed 0-1 quadratic program by replacing (\ref{eq:1}) with (\ref{eq:Li1}) and adding the set of constraints (\ref{eq:Li2}) and (\ref{eq:Li3}).

\vspace{-.25cm}
{\footnotesize
\begin{flalign}
\mbox{max} & \displaystyle\sum_{i=1}^n \left( \displaystyle\sum_{k=1}^{m_i} p_{ik} x_{ik} - Rd_i y_i\right)\label{eq:Li1}
\\ 
\mbox{s.t.} \;& \upsilon_{max}y_i + \nu \displaystyle\sum_{j=1}^i \displaystyle\sum_{k=1}^{m_j} w_{jk} x_{jk} y_i = 1, \; i=1,\ldots,n \label{eq:Li2}
\\ 
& y_i \in \mathbb{R}_+, \; i=1,\ldots,n \label{eq:Li3}
\end{flalign}
}
\vspace{-.25cm}

If $z=xy$ is a polynomial mixed 0-1 term where $x$ is binary and $y$ is a real variable, then it can be linearized via the set of linear inequalities: (i) $z\leq Ux$;  (ii) $z\geq Lx$;  (iii) $z\leq y + L\left(x-1\right)$; (iiii) $z\geq y + U\left(x-1\right)$ (see \cite{Tawarmalani02}). $U$ and $L$ are the upper and lower bounds on $y$, i.e. $L\leq y\leq U$. We can linearize the $x_{jk} y_i$ term in (\ref{eq:Li2}) by introducing a new real variable $z^i_{jk} = x_{jk} y_i$. Furthermore, let $p_i^c$ and $w_i^c$ denote the total profit and the total weight of the compulsory items in  city $i$ according to Proposition~\ref{prop2}. Variable $y_i$, $i=1, \ldots, n$, can be bounded from below by $L_i = 1 / \left(\upsilon_{max}- \nu\sum_{j=1}^i w_j^c\right)$. Similarly, let $w_i^{max}$ be the total weight of the items (including all the compulsory items) in city $i$. We can bound $y_i$, $i=1, \ldots, n$, from above by $U_i = 1 / \left(\upsilon_{max}- \nu \cdot min\left(\sum_{j=1}^i w_j^{max},W\right)\right)$ and formulate {\NKPc} as the following linear mixed 0-1 program ({\ENKP}):

\vspace{-.25cm}
{\footnotesize
\begin{flalign}
\nonumber\mbox{max} & \displaystyle\sum_{i=1}^n \left(p_i^c +\displaystyle\sum_{k=1}^{m_i} p_{ik} x_{ik} - R d_{i} y_i\right)
\\
\nonumber\mbox{s.t.} \, &\upsilon_{max}y_i + \nu \left(w_i^c + \displaystyle\sum_{j=1}^i \displaystyle\sum_{k=1}^{m_j} w_{jk} z^i_{jk} \right)= 1, \; i=1,\ldots,n
\\
\nonumber&z^i_{jk}\leq U_ix_{jk}, \; i,j=1,\ldots,n,\; j\leq i,\; e_{jk} \in M_{j}
\\
\nonumber&z^i_{jk}\geq L_ix_{jk}, \; i,j=1,\ldots,n,\; j\leq i,\; e_{jk} \in M_{j}
\\
\nonumber&z^i_{jk}\geq y_i + U_i\left(x_{jk}-1\right), \; i,j=1,\ldots,n, \; j\leq i,\; e_{jk} \in M_{j}
\\
\nonumber&z^i_{jk}\leq y_i + L_i\left(x_{jk}-1\right), \; i,j=1,\ldots,n, \; j\leq i,\; e_{jk} \in M_{j}
\\ 
&\displaystyle\sum_{i=1}^n \displaystyle\sum_{k=1}^{m_i} w_{ik} x_{ik} \leq W \label{eq:e6}
\\  
\nonumber&x_{ik} \in \left\{0,1\right\}, \; e_{ik} \in M
\\ 
\nonumber&z^i_{jk}\in \mathbb{R}_+, \; i,j=1,\ldots,n, \; j\leq i,\; e_{jk} \in M_{j}
\\ 
\nonumber&y_i \in \mathbb{R}_+, \; i=1,\ldots,n
\end{flalign}}
\vspace{-.25cm}

We now introduce a set of inequalities in order to obtain tighter relaxations. The reformulation-linearization technique by \cite{Sherali99} uses $3n$ additional inequalities for the capacity constraint (\ref{eq:e6}). Multiplying (\ref{eq:e6}) by $y_l$, $U_l-y_l$ and $y_l-L_l$, $l=1, \ldots, n$, we obtain the inequalities

\vspace{-.25cm}
{\footnotesize
\begin{flalign}
\nonumber&\displaystyle\sum_{i=1}^n \displaystyle\sum_{k=1}^{m_i} w_{ik} z^l_{ik} \leq Wy_l;
\\  
\nonumber&U_l\displaystyle\sum_{i=1}^n \displaystyle\sum_{k=1}^{m_i} w_{ik} x_{ik}- \displaystyle\sum_{i=1}^n \displaystyle\sum_{k=1}^{m_i} w_{ik} z^l_{ik} \leq U_lW - Wy_l;
\\ 
\nonumber&\displaystyle\sum_{i=1}^n \displaystyle\sum_{k=1}^{m_i} w_{ik} z^l_{ik} - L_l\displaystyle\sum_{i=1}^n \displaystyle\sum_{k=1}^{m_i} w_{ik} x_{ik} \leq Wy_l - L_lW.
\end{flalign}}
\vspace{-.25cm}

Another set of inequalities can be derived from the fact that the item $e_{il}$ in the city $i$ should not be selected if in the same city there exists unselected item $e_{ik}$ with $p_{il}<p_{ik}$ and $w_{il}>w_{ik}$. 
Furthermore, the item $e_{jl}$ in the city $j$ should not be selected if there exists unselected item $e_{ik}$ in the city $i$, with $j<i$, $p_{jl}-\Delta_l^{ji}<p_{ik}$ and $w_{jl}>w_{ik}$ where
{\footnotesize
\begin{equation}
\nonumber \Delta_l^{ji}=R \sum_{a=j}^{i-1} d_a\left(\frac{1}{\upsilon_{max}-\nu \left(w_{jl}+\sum_{b=1}^{a} w_b^c\right)}-\frac{1}{\upsilon_{max}-\nu \sum_{b=1}^{a} w_b^c}\right)
\end{equation}}
\noindent is a lower bound on the transportation cost to deliver $e_{jl}$ from $j$ to $i$. Similarly, the item $e_{ik}$ in the city $i$ should not be selected if there exists unselected item $e_{jl}$ in the city $j$, with $j<i$, $p_{jl}-\overline{\Delta}_l^{ji}>p_{ik}$ and $w_{jl}<w_{ik}$ where
{\footnotesize
\begin{equation}
\nonumber \overline{\Delta}_l^{ji}=R\sum_{a=j}^{i-1} d_a \left(\frac{1}{\upsilon_{max}-\nu\cdot min\left(w_{jl}+\displaystyle\sum_{b=1}^{a} w_b^{max},W\right)}-\frac{1}{\upsilon_{max}-\nu\cdot min\left(\displaystyle\sum_{b=1}^{a} w_b^{max},W\right)}\right)
\end{equation}
}
\noindent is an upper bound on the transportation cost to deliver $e_{jl}$ from $j$ to $i$. This leads to the following inequalities for $i,j=1, \ldots, n$:

\vspace{-.25cm}
{\footnotesize
{\begin{flalign}
&x_{il}\leq x_{ik}, \; e_{il},e_{ik} \in M_i \;:\; l\neq k, \; p_{il}<p_{ik},\; w_{il}>w_{ik}; \label{eq:ve1}
\\  
&x_{jl}\leq x_{ik}, \; j<i, \; e_{jl}\in M_j, \; e_{ik} \in M_i, \;: \; p_{jl}-\Delta_l^{ji}<p_{ik}, \; w_{jl}>w_{ik};\label{eq:ve2}
\\  
&x_{jl}\geq x_{ik}, \; j<i, \; e_{jl}\in M_j, \; e_{ik} \in M_i, \;: \; p_{jl}-\overline{\Delta}_l^{ji}>p_{ik}, \; w_{jl}<w_{ik}.\label{eq:ve3}
\end{flalign}
}
\vspace{-.25cm}

\section{Approximate Solution} \label{sec:AS}
In practice, the use of approximations is an efficient way to deal with nonlinear terms. Although the approximate solution is likely to be different from the exact one, it might be close enough and obtainable in a reasonable computational time. 

Consider an arbitrary pair $\left(i,i+1\right)$, $i=1, \ldots, n$, and the traveling time $t'_i \in [t_{\min}, t_{\max}]$ per distance unit  along it. Here $t_{max}=1/\upsilon_{min}$ and $t_{min}=1/\upsilon_{max}$ denote the maximum and minimum travel time per unit, respectively. 
We partition the interval $\left[t_{min},t_{max}\right]$ into $\tau$ equal-sized sub-intervals and determine thus a set $T=\left\{T_1,\ldots,T_{\tau}\right\}$ of straight line segments to approximate the curve $t\left(\upsilon\right)$ as illustrated in Figure~\ref{fig:fig1}. Each segment $a \in T$ is characterized by its minimal velocity $\upsilon_a^{min}$ and its corresponding maximum traveling time per distance unit $t_a^{max}$, and by its maximum velocity $\upsilon_a^{max}$ and its corresponding minimum traveling time per distance unit $t_a^{min}$. Specifically, $\left(\upsilon_a^{min},t_a^{max}\right)$ and $\left(\upsilon_a^{max},t_a^{min}\right)$ are the endpoints of $a \in T$ referred to as breakpoints. We approximate $t'_i$ by the linear combination of $t_a^{max}$ and $t_a^{min}$ if $\upsilon_i \in \left[\upsilon_a^{min},\upsilon_a^{max}\right]$. 

\begin{figure}[htb]
\centering
\includegraphics[height=5cm]{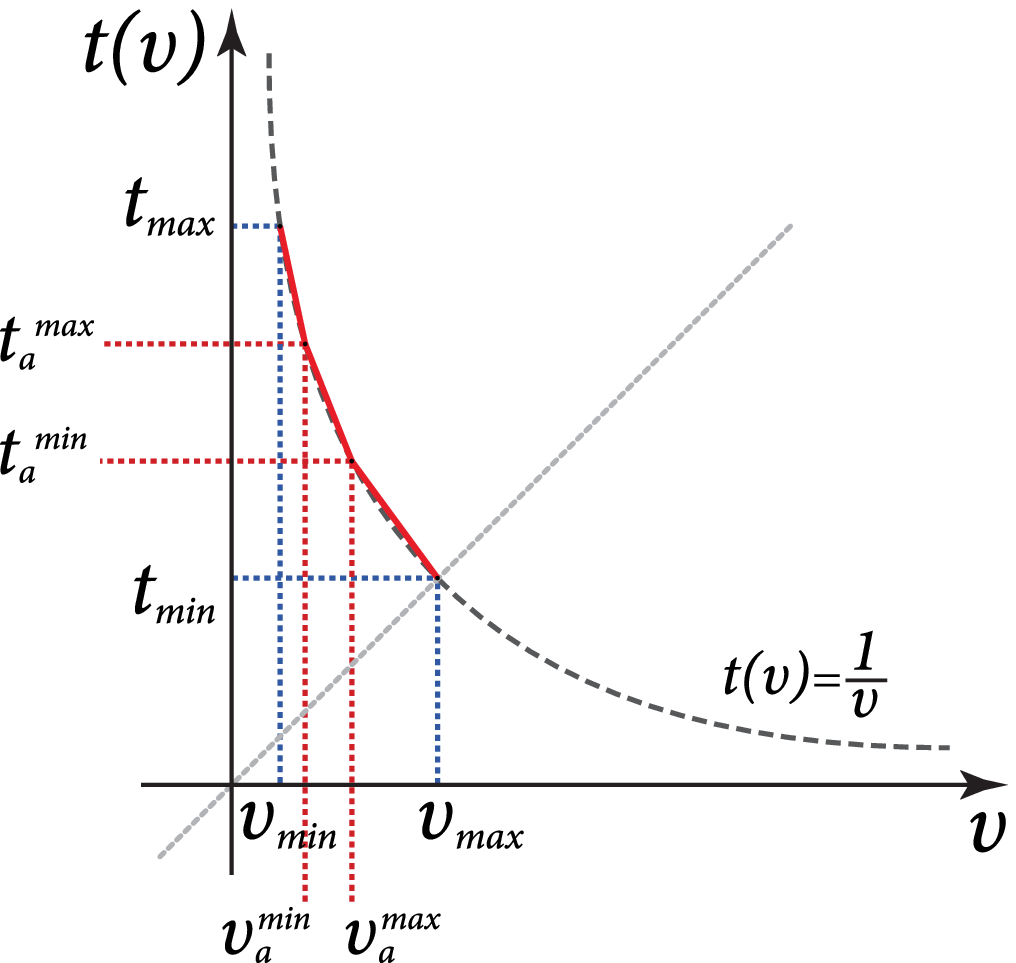}
\caption{Piecewise linear approximation of $t\left(\upsilon\right)=1/\upsilon$}
\label{fig:fig1}
\end{figure}

Our approximation model uses three types of decision variables in addition to the binary variable $x_{ik}$ for each item $e_{ik} \in M$ from Section~\ref{sec:prob}. Let $w_i$ be a real variable equal to the total weight of selected items when traveling along the $\left(i,i+1\right)$. Let $p_i$ be a real variable equal to the difference of the total profit of selected items and their total transportation costs when delivering them to city $i+1$. We set $w_0=p_0=0$. Let $A_i \subseteq T$, $1 \leq i \leq n$, denote a set of possible segments to which velocity $\upsilon_i$ of the vehicle may belong, i.e. $A_i = \left\{a \in T\; :\;
\left( \upsilon_a^{\min} \in \left[ \upsilon_i^{\min}, \upsilon_i^{\max} \right] \right) \vee 
\left( \upsilon_a^{\max} \in \left[ \upsilon_i^{\min}, \upsilon_i^{\max} \right] \right) \right\}$, where $\upsilon_i^{\max} = \upsilon_{\max}- \nu\sum_{j=1}^i w_j^c$ is the maximal possible velocity that the vehicle can move along $\left(i,i+1\right)$ when packing in all compulsory items only, and $\upsilon_i^{\min}=\upsilon_{max}- \nu \cdot \min\left(\sum_{j=1}^i w_j^{max},W\right)$ the minimum possible velocity along $(i, i+1)$ after having packed in all items available in cities $1, \ldots, i$. Actually, we have $\upsilon_i \in \left[\upsilon_i^{min}, \upsilon_i^{max}\right]$.

When $\upsilon_i \in \left[ \upsilon_a^{\min}, \upsilon_a^{\max} \right]$ for $a \in T$, any point in between endpoints of $a$ is a weighted sum of them. Let $B_i$ denote a set of all breakpoints that the linear segments of $A_i$ have. Then the value of the real variable $y_{ib} \in \left[0,1\right]$ is a weight assigned to the breakpoint $b \in B_i$, $b \sim \left(\upsilon_b,t_b\right)$. {\NKPc} (and {\NKPu}) can be approximated by the following linear mixed 0-1 program ({\ANKP}):

\vspace{-.25cm}
{\footnotesize
\begin{flalign}
\mbox{max} \;  & p_n \label{eq:a0}
\\
\mbox{s.t.} \; & p_i=p_{i-1}+p_i^c+\displaystyle\sum_{e_{ik}\in M_i} p_{ik}x_{ik}-Rd_i\displaystyle\sum_{b\in B_i} t_by_{ib}, \; i=1,\ldots,n \label{eq:a1}
\\
& w_i=w_{i-1} + w_i^c+\displaystyle\sum_{e_{ik}\in M_i} w_{ik}x_{ik}, \; i=1,\ldots,n \label{eq:a2}
\\
& \nu w_i+\displaystyle\sum_{b\in B_i} \upsilon_b y_{ib}=\upsilon_{max}, \; i=1,\ldots,n \label{eq:a3}
\\
& \displaystyle\sum_{b\in B_i} y_{ib} = 1, \; i=1,\ldots,n \label{eq:a4}
\\
& w_n\leq W \label{eq:a5}
\\  
& x_{ik} \in \left\{0,1\right\}, \; e_{ik} \in M \label{eq:a6}
\\  
& y_{ib} \in \left[0,1\right], \; i=1,\ldots,n,\;b \in B_i \label{eq:a7}
\\ 
& p_i\in \mathbb{R}, \; i=1,\ldots,n \label{eq:a8}
\\
& w_i\in \mathbb{R}_{\geq 0}, \; i=1,\ldots,n \label{eq:a9}
\\
& w_0=p_0=0
\end{flalign}
}
\vspace{-.25cm}

Equation (\ref{eq:a0}) defines the objective $p_n$ as the difference of the total profit of selected items and their total transportation costs delivered to city $n+1$. Since the transportation costs are approximated in {\ANKP}, the actual objective value for {\NKPc} (and {\NKPu}) should be computed on values of decision variables of vector $X$. Equation (\ref{eq:a1}) computes the difference $p_i$ of the total profit of selected items and their total transportation costs when arriving at city $i+1$ by summing up the value of $p_{i-1}$ concerning $\left(i-1,i\right)$, the profit of compulsory items $p_i^c$ and the profit $\sum_{e_{ik}\in M_i} p_{ik}x_{ik}$ of items selected in city $i$, and subtracting the approximated transportation costs along $(i, i+1)$. Equation (\ref{eq:a2}) gives the weight $w_i$ of the selected items when the vehicle departs city $i$ by summing up $w_{i-1}$, the weight of compulsory items $w_i^c$ and the weight $\sum_{e_{ik}\in M_i} w_{ik}x_{ik}$ of items selected in city $i$. Equation (\ref{eq:a3}) implicitly defines the segment $a \in A_i$ to which the velocity of the vehicle $\upsilon_i$ belongs and sets the weights of its breakpoints. Equation (\ref{eq:a4}) forces the total weight of the breakpoints of $B_i$ be exactly 1. Equation (\ref{eq:a5}) imposes the capacity constraint, and Eq. (\ref{eq:a6}) declares $x_{ik}$ as binary. Equation (\ref{eq:a7}) states $y_{ib}$ as a real variable defined in $\left[0,1\right]$. Finally, Equation (\ref{eq:a8}) declares $p_i$ as a real variable, while Eq. (\ref{eq:a9}) defines $w_i$ as a non-negative real. A solution of {\ANKP} can be used as a starting solution for {\ENKP} in the case that all sets of inequalities (\ref{eq:ve1}), (\ref{eq:ve2}) and (\ref{eq:ve3}) are met.

\section{Computational Experiments} \label{sec:CE}

\begin{table}[!htbp]
\centering
\caption{Results of Computational Experiments on Small Size Instances}
\label{tab:res1}
{\tiny
\begin{tabular}{r||r|r|c||r|r||r|r|r||r|r|r}
\hline
instance & $m$ & $\alpha$ & {\textit{ver}} & $t^e$ & $gap^e$ & $\rho^{100}$ & $t^{100}$ & $\beta^{100}$ & $\rho^{1000}$ & $t^{1000}$ & $\beta^{1000}$\\
\hline
\hline
\multicolumn{12}{c}{instance family \texttt{eil51}} \\
\hline
uncorr\_01 & 50 & 42.0 & c & 1 & 0.00 & 1.0000 & 0 & 56.9 & 1.0000 & 1 & 55.9 \\
uncorr\_06 & 50 & 14.0 & c & 3 & 0.00 & 1.0000 & 0 & 39.9 & 1.0000 & 0 & 38.7 \\
uncorr\_10 & 50 & 50.0 & u & 1 & 0.00 & 1.0000 & 0 & 11.3 & 1.0000 & 0 & 9.4 \\
uncorr-s-w\_01 & 50 & 30.0 & c & 0 & 0.00 & 1.0000 & 0 & 79.0 & 1.0000 & 1 & 78.0 \\
uncorr-s-w\_06 & 50 & 24.0 & c & 3 & 0.00 & 1.0000 & 0 & 36.5 & 1.0000 & 0 & 35.2 \\
uncorr-s-w\_10 & 50 & 34.0 & u & 3 & 0.00 & 1.0000 & 0 & 13.4 & 1.0000 & 0 & 11.9 \\
b-s-corr\_01 & 50 & 4.0 & c & 2 & 0.00 & 1.0000 & 0 & 91.5 & 1.0000 & 2 & 90.5 \\
b-s-corr\_06 & 50 & 0.0 & c & 249 & 0.00 & 1.0000 & 0 & 54.5 & 1.0000 & 1 & 53.3 \\
b-s-corr\_10 & 50 & 0.0 & c & 139 & 0.00 & 1.0000 & 0 & 26.2 & 1.0000 & 0 & 24.9 \\
uncorr\_01 & 250 & 39.2 & c & 1855 & 0.00 & 1.0000 & 0 & 66.8 & 1.0000 & 1 & 65.7 \\
uncorr\_06 & 250 & 16.4 & c & - & 10.66 & 1.0000 & 0 & 39.0 & 1.0000 & 0 & 37.8 \\
uncorr\_10 & 250 & 54.4 & u & 268 & 0.00 & 1.0000 & 0 & 11.2 & 1.0000 & 0 & 9.5 \\
uncorr-s-w\_01 & 250 & 20.8 & c & 22 & 0.00 & 1.0000 & 0 & 89.8 & 1.0000 & 1 & 88.8 \\
uncorr-s-w\_06 & 250 & 14.0 & c & - & 25.20 & 1.0000 & 0 & 45.5 & 1.0000 & 0 & 44.2 \\
uncorr-s-w\_10 & 250 & 19.2 & u & 73472 & 0.00 & 1.0000 & 0 & 16.0 & 1.0000 & 0 & 14.6 \\
b-s-corr\_01 & 250 & 0.0 & c & - & 0.89 & 1.0000 & 0 & 92.0 & 1.0000 & 1 & 91.1 \\
b-s-corr\_06 & 250 & 0.0 & c & - & 53.48 & 1.0000 & 0 & 56.9 & 1.0000 & 1 & 55.7 \\
b-s-corr\_10 & 250 & 0.0 & c & - & 60.94 & 1.0000 & 0 & 27.3 & 1.0000 & 0 & 25.9 \\
uncorr\_01 & 500 & 37.0 & c & - & 14.82 & 1.0000 & 0 & 69.1 & 1.0000 & 1 & 68.0 \\
uncorr\_06 & 500 & 15.2 & c & - & 21.26 & 1.0000 & 0 & 39.6 & 1.0000 & 0 & 38.3 \\
uncorr\_10 & 500 & 51.4 & u & - & 1.27 & 1.0000 & 0 & 11.8 & 1.0000 & 0 & 10.1 \\
uncorr-s-w\_01 & 500 & 20.2 & c & - & 1.80 & 1.0000 & 0 & 90.8 & 1.0000 & 1 & 89.9 \\
uncorr-s-w\_06 & 500 & 15.2 & c & - & 37.83 & 0.9999 & 0 & 45.1 & 1.0000 & 0 & 43.9 \\
uncorr-s-w\_10 & 500 & 18.6 & u & - & 4.44 & 1.0000 & 0 & 16.4 & 1.0000 & 0 & 15.0 \\
b-s-corr\_01 & 500 & 0.0 & c & - & 5.97 & 1.0000 & 0 & 93.1 & 1.0000 & 2 & 92.1 \\
b-s-corr\_06 & 500 & 0.0 & c & - & 49.28 & 1.0000 & 0 & 56.5 & 1.0000 & 0 & 55.4 \\
b-s-corr\_10 & 500 & 0.0 & c & - & 71.87 & 1.0000 & 0 & 26.6 & 1.0000 & 0 & 25.2 \\
\hline 
\hline
\multicolumn{12}{c}{instance family \texttt{eil76}} \\
\hline
uncorr\_01 & 75 & 26.7 & c & 4 & 0.00 & 1.0000 & 0 & 77.7 & 1.0000 & 1 & 76.7 \\
uncorr\_06 & 75 & 14.7 & c & 50 & 0.00 & 1.0000 & 0 & 34.3 & 1.0000 & 0 & 33.1 \\
uncorr\_10 & 75 & 48.0 & u & 15 & 0.00 & 1.0000 & 0 & 11.5 & 1.0000 & 0 & 9.6 \\
uncorr-s-w\_01 & 75 & 26.7 & c & 1 & 0.00 & 1.0000 & 0 & 79.2 & 1.0000 & 3 & 78.2 \\
uncorr-s-w\_06 & 75 & 17.3 & c & 82 & 0.00 & 1.0000 & 0 & 41.3 & 1.0000 & 1 & 40.1 \\
uncorr-s-w\_10 & 75 & 16.0 & u & 9 & 0.00 & 1.0000 & 0 & 16.8 & 1.0000 & 0 & 15.4 \\
b-s-corr\_01 & 75 & 0.0 & c & 6 & 0.00 & 1.0000 & 0 & 94.7 & 1.0000 & 1 & 93.8 \\
b-s-corr\_06 & 75 & 0.0 & c & - & 8.53 & 1.0000 & 0 & 59.7 & 1.0000 & 2 & 58.5 \\
b-s-corr\_10 & 75 & 0.0 & c & 4555 & 0.00 & 1.0000 & 0 & 25.9 & 1.0000 & 0 & 24.5 \\
uncorr\_01 & 375 & 38.1 & c & - & 15.49 & 1.0000 & 0 & 67.2 & 1.0000 & 2 & 66.1 \\
uncorr\_06 & 375 & 16.0 & c & - & 18.04 & 1.0000 & 0 & 37.5 & 1.0000 & 0 & 36.2 \\
uncorr\_10 & 375 & 49.3 & u & - & 0.57 & 1.0000 & 0 & 12.0 & 1.0000 & 0 & 10.2 \\
uncorr-s-w\_01 & 375 & 14.9 & c & 30376 & 0.00 & 1.0000 & 0 & 90.9 & 1.0000 & 5 & 89.9 \\
uncorr-s-w\_06 & 375 & 12.3 & c & - & 48.36 & 1.0000 & 0 & 47.4 & 1.0000 & 1 & 46.2 \\
uncorr-s-w\_10 & 375 & 14.9 & u & - & 3.70 & 1.0000 & 0 & 17.3 & 1.0000 & 0 & 15.9 \\
b-s-corr\_01 & 375 & 0.0 & c & - & 9.32 & 1.0000 & 0 & 95.4 & 1.0000 & 2 & 94.4 \\
b-s-corr\_06 & 375 & 0.0 & c & - & 60.98 & 1.0000 & 0 & 57.4 & 1.0000 & 1 & 56.2 \\
b-s-corr\_10 & 375 & 0.0 & c & - & 69.90 & 1.0000 & 0 & 27.8 & 1.0000 & 1 & 26.6 \\
uncorr\_01 & 750 & 32.5 & c & - & 19.52 & 1.0000 & 0 & 72.3 & 1.0000 & 2 & 71.2 \\
uncorr\_06 & 750 & 14.8 & c & - & 33.14 & 1.0000 & 0 & 39.5 & 1.0000 & 0 & 38.3 \\
uncorr\_10 & 750 & 43.1 & u & - & 5.25 & 1.0000 & 0 & 13.1 & 1.0000 & 0 & 11.4 \\
uncorr-s-w\_01 & 750 & 16.7 & c & - & 11.31 & 1.0000 & 0 & 89.8 & 1.0000 & 2 & 88.9 \\
uncorr-s-w\_06 & 750 & 13.5 & c & - & 60.27 & 1.0000 & 0 & 46.3 & 1.0000 & 1 & 45.1 \\
uncorr-s-w\_10 & 750 & 14.4 & u & - & 6.88 & 1.0000 & 0 & 17.2 & 1.0000 & 0 & 15.9 \\
b-s-corr\_01 & 750 & 0.0 & c & - & 10.46 & 1.0000 & 0 & 95.0 & 1.0000 & 2 & 94.0 \\
b-s-corr\_06 & 750 & 0.0 & c & - & 62.42 & 1.0000 & 0 & 56.1 & 1.0000 & 1 & 54.9 \\
b-s-corr\_10 & 750 & 0.0 & c & - & 84.45 & 1.0000 & 0 & 26.2 & 1.0000 & 0 & 24.9 \\
\hline 
\hline
\multicolumn{12}{c}{instance family \texttt{eil101}} \\
\hline
uncorr\_01 & 100 & 49.0 & c & 9 & 0.00 & 1.0000 & 0 & 61.3 & 1.0000 & 1 & 60.2 \\
uncorr\_06 & 100 & 16.0 & c & 714 & 0.00 & 0.9999 & 0 & 40.1 & 1.0000 & 2 & 38.8 \\
uncorr\_10 & 100 & 57.0 & u & 21 & 0.00 & 1.0000 & 0 & 10.2 & 1.0000 & 0 & 8.5 \\
uncorr-s-w\_01 & 100 & 25.0 & c & 3 & 0.00 & 1.0000 & 0 & 91.2 & 1.0000 & 1 & 90.3 \\
uncorr-s-w\_06 & 100 & 17.0 & c & 446 & 0.00 & 1.0000 & 0 & 42.3 & 1.0000 & 1 & 41.0 \\
uncorr-s-w\_10 & 100 & 15.0 & u & 68 & 0.00 & 1.0000 & 0 & 17.4 & 1.0000 & 0 & 16.0 \\
b-s-corr\_01 & 100 & 0.0 & c & 532 & 0.00 & 1.0000 & 0 & 95.4 & 1.0000 & 4 & 94.4 \\
b-s-corr\_06 & 100 & 0.0 & c & - & 44.03 & 1.0000 & 0 & 56.8 & 1.0000 & 2 & 55.7 \\
b-s-corr\_10 & 100 & 0.0 & c & - & 28.96 & 0.9999 & 0 & 28.5 & 1.0000 & 1 & 27.2 \\
uncorr\_01 & 500 & 38.8 & c & - & 13.92 & 1.0000 & 0 & 66.6 & 1.0000 & 3 & 65.5 \\
uncorr\_06 & 500 & 14.4 & c & - & 20.49 & 1.0000 & 0 & 39.6 & 1.0000 & 1 & 38.4 \\
uncorr\_10 & 500 & 51.4 & u & - & 1.94 & 1.0000 & 0 & 11.5 & 1.0000 & 0 & 9.8 \\
uncorr-s-w\_01 & 500 & 20.4 & c & - & 7.00 & 1.0000 & 1 & 89.3 & 1.0000 & 14 & 88.3 \\
uncorr-s-w\_06 & 500 & 14.2 & c & - & 40.92 & 1.0000 & 0 & 45.3 & 1.0000 & 1 & 44.1 \\
uncorr-s-w\_10 & 500 & 16.4 & u & - & 7.20 & 1.0000 & 0 & 16.4 & 1.0000 & 0 & 15.1 \\
b-s-corr\_01 & 500 & 0.0 & c & - & 13.73 & 1.0000 & 1 & 94.4 & 1.0000 & 3 & 93.5 \\
b-s-corr\_06 & 500 & 0.0 & c & - & 68.68 & 1.0000 & 0 & 55.3 & 1.0000 & 2 & 54.1 \\
b-s-corr\_10 & 500 & 0.0 & c & - & 77.57 & 1.0000 & 0 & 26.3 & 1.0000 & 0 & 25.1 \\
uncorr\_01 & 1000 & 37.0 & c & - & 26.74 & 0.9999 & 0 & 67.2 & 1.0000 & 3 & 66.1 \\
uncorr\_06 & 1000 & 15.1 & c & - & 30.91 & 1.0000 & 0 & 39.5 & 1.0000 & 1 & 38.3 \\
uncorr\_10 & 1000 & 50.4 & u & - & 4.69 & 1.0000 & 0 & 11.8 & 1.0000 & 0 & 10.1 \\
uncorr-s-w\_01 & 1000 & 19.7 & c & - & 10.46 & 0.9999 & 248 & 89.3 & 1.0000 & 6144 & 88.3 \\
uncorr-s-w\_06 & 1000 & 13.7 & c & - & 57.02 & 1.0000 & 0 & 45.6 & 1.0000 & 1 & 44.4 \\
uncorr-s-w\_10 & 1000 & 15.9 & u & - & 13.54 & 1.0000 & 0 & 16.7 & 1.0000 & 0 & 15.3 \\
b-s-corr\_01 & 1000 & 0.0 & c & - & 14.41 & 1.0000 & 1 & 93.9 & 1.0000 & 7 & 93.0 \\
b-s-corr\_06 & 1000 & 0.0 & c & - & 80.39 & 1.0000 & 0 & 55.8 & 1.0000 & 2 & 54.6 \\
b-s-corr\_10 & 1000 & 0.0 & c & - & 97.54 & 1.0000 & 0 & 27.1 & 1.0000 & 1 & 25.8 \\
\hline
\end{tabular}
}
\end{table}

We now investigate the effectiveness of proposed approaches by experimental studies. On the one hand, we evaluate our MIP models {\ENKP} and {\ANKP} in terms of solution quality and running time. On the other hand, we assess the advantage of the pre-processing scheme in terms of quantity of discarded items and auxiliary decision variables. The program code is implemented in JAVA using the \textsc{Cplex} 12.6 library with default settings. The experiments have been carried out on a computational cluster with 128 Gb RAM and 2.8 GHz 48-cores AMD Opteron processor. 

The test instances are adopted from the benchmark set $B$ of \cite{Polyakovskiy14}.
This benchmark set is constructed on TSP instances from TSPLIB (see \cite{Reinelt91}). In addition, it contains for each city but the first one a set of items. We use the set of items available at each city and obtain the route from the corresponding TSP instance by running the Chained Lin-Kernighan heuristic (see \cite{chainedLK03applegate}).
Given the permutation $\pi = (\pi_1, \pi_2, \ldots, \pi_n)$ of cities computed by the Chained Lin-Kernighan heuristic, where $\pi_1$ is free of items, we use $N=(\pi_2, \pi_3, \ldots, \pi_n, \pi_1)$ as the route for our problem. We consider the uncorrelated, uncorrelated with similar weights, and bounded strongly correlated types of items' generation, and set $\upsilon_{min}$ and $\upsilon_{max}$ to 0.1 and 1 as proposed for $B$.

The results of our experiments are shown in Tables \ref{tab:res1} and \ref{tab:res2}. First, we investigate three families of small size instances based on the TSP problems \texttt{eil51}, \texttt{eil76}, and \texttt{eil101} with $51$, $76$ and $101$ cities, respectively. Note that all instances of a family have the same route $N$. We considered instances with $1$, $5$, and $10$ items per city. The postfixes $1$, $6$ and $10$ in the instances' names indicate the capacity $W$.
Column $2$ specifies the total number of items $m$. A ratio $\alpha=100\left(m-m'\right)/m$ in Column $3$ denotes a percentage of items discarded in pre-processing step, where $m'$ is the number of items left after pre-processing. Column $4$ identifies by \textit{``u''} whether {\NKPc} has been reduced to {\NKPu} by pre-processing. Columns $5$ and $6$ report a computational time in seconds and a relative gap reached in percents for {\ENKP}. The time limit of 1 day has been given to {\ENKP}. Thus, Column $5$ either contains a required time or ``-'' if the time limit is reached. Results for {\ANKP} with $\tau=100$ are demonstrated in Columns $7$ and $8$, while the case of $\tau=1000$ is shown in Columns $10$ and $11$. Columns $7$ and $10$ report $\rho^\tau$ as a ratio between the best lower bounds obtained by {\ANKP} and {\ENKP}. Within the experiments, {\ANKP} with $\tau=100$ produces an initial solution for {\ENKP}. Columns $8$ and $11$ contain running times of {\ANKP}. The time limit of 2 hours has been given to {\ANKP}. Finally, columns $9$ and $12$ show a rate $\beta^\tau$ which is a percentage of auxiliary decision variables $y_{ib}$ for $i=1,\ldots,n$ and $b\in {B_i}$ used in practice by {\ANKP}. At most $\tau n$ variables is required by {\ANKP}. Thus, $\beta$ is computed as $\beta=100\left(\sum_{i=1}^n \left|B_i\right|\right)/\left(\tau n\right)$.

The results show that only the instances of small size are solved by {\ENKP} to optimality within the given time limit. At the same time, the unconstrained cases of the problem turn out to be easier to handle. They either are solved to optimality or have a low relative gap comparing to the constrained cases, even when latter have less number of items $m$. Generally, the instances with large $W$ are liable to reduction.  Because $W$ is large, they have more chances to loose enough items so that the total weight of rest items becomes less or equal to $W$. However, the pre-processing scheme does not work for bounded strongly correlated type of the instances. No instance of this type is reduced to {\NKPu}. Moreover, the results show that this type is presumably harder to solve comparing to others as expected in \cite{Polyakovskiy14}. In fact, the relative gap is significantly larger concerning this type. 

{\ANKP} is particularly fast and its model is solved to optimality for all the small size instances within the given time limit. The ratio $\rho^\tau$ is very close to $1$ which leads to two observations.  Firstly, {\ANKP} obtains approximately the same result as the optimal solution of {\ENKP} has but in a shorter time. Secondly, {\ENKP} cannot find much better solutions even within large given time. Therefore, we can conclude that {\ANKP} gives an advanced trade-off in terms of computational time and solution's quality comparing to {\ENKP}. It looks very swift even with instances of hard bounded strongly correlated type. Moreover, {\ANKP} produces very good approximation even for reasonably small $\tau=100$. Only one instance of the whole test suite causes a considerable difficulty for {\ANKP} in terms of a running time. The rate $\beta^\tau$ demonstrates that in practice {\ANKP} uses a very reduced set of auxiliary decision variables. The medians over all entries of $\beta^{100}$ and $\beta^{1000}$ are 45.3 and 44.1, respectively. In general, $\beta^\tau$ is significantly small when $W$ is large, since latter results in a slower growth of diapason $\left[\upsilon_i^{\min},\upsilon_i^{\max}\right]$ in {\ANKP}, for $i=1,\ldots,n$. In other words, the instances with large $W$ require less number of auxiliary decision variables comparing to the instances where $W$ is smaller.

The goal of our second experiment is to understand how fast {\ANKP} handles instances of larger size. We use the same settings as for the first experiment, but now give {\ANKP} the time limit of 6 hours and set $\tau=100$. We investigate two families of largest size instances of $B$ of \cite{Polyakovskiy14}, namely those based on the TSP problems \texttt{pla33810} and \texttt{pla85900} with $33810$ and $85900$ cities, respectively. Table \ref{tab:res2} reports the results. {\ANKP} needs less than $\sim 40$ minutes to solve any instance of family \texttt{pla33810}. Almost all instances of family \texttt{pla85900} can be solved within 2 hours; it takes no longer than $\sim 5.5$ hours for any of them. Therefore, {\ANKP} proves its ability to master large problems in a reasonable time.

\begin{table}[!htbp]
\centering
\caption{Results of Computational Experiments on Large Size Instances}
\label{tab:res2}
{\tiny
\begin{tabular}{l@{\hskip 0.5in}l}
\begin{tabular}{r||r|r|c||r|r}
\hline
instance & $m$ & $\alpha$ & {\textit{ver}} & $t^{100}$ & $\beta^{100}$\\
\hline
\hline
\multicolumn{6}{c}{instance family \texttt{pla33810}} \\
\hline
uncorr\_01 & 33809 & 29.0 & c & 522 & 77.7 \\
uncorr\_06 & 33809 & 12.8 & c & 337 & 41.9 \\
uncorr\_10 & 33809 & 35.9 & u & 32 & 14.2 \\
uncorr-s-w\_01 & 33809 & 19.3 & c & 425 & 88.5 \\
uncorr-s-w\_06 & 33809 & 11.2 & c & 634 & 46.8 \\
uncorr-s-w\_10 & 33809 & 8.7 & c & 33 & 17.3 \\
b-s-corr\_01 & 33809 & 0.0 & c & 419 & 92.6 \\
b-s-corr\_06 & 33809 & 0.0 & c & 582 & 55.3 \\
b-s-corr\_10 & 33809 & 0.0 & c & 696 & 25.6 \\
uncorr\_01 & 169045 & 30.6 & c & 601 & 75.6 \\
uncorr\_06 & 169045 & 12.8 & c & 1276 & 41.8 \\
uncorr\_10 & 169045 & 35.8 & u & 72 & 13.9 \\
uncorr-s-w\_01 & 169045 & 15.2 & c & 389 & 89.5 \\
uncorr-s-w\_06 & 169045 & 11.7 & c & 600 & 46.3 \\
uncorr-s-w\_10 & 169045 & 9.0 & c & 774 & 17.1 \\
b-s-corr\_01 & 169045 & 0.0 & c & 1526 & 92.7 \\
b-s-corr\_06 & 169045 & 0.0 & c & 433 & 55.4 \\
b-s-corr\_10 & 169045 & 0.0 & c & 830 & 25.4 \\
uncorr\_01 & 338090 & 31.6 & c & 2079 & 74.5 \\
uncorr\_06 & 338090 & 12.8 & c & 1272 & 41.7 \\
uncorr\_10 & 338090 & 35.9 & u & 1264 & 13.8 \\
uncorr-s-w\_01 & 338090 & 15.2 & c & 1266 & 89.6 \\
uncorr-s-w\_06 & 338090 & 11.9 & c & 1225 & 46.2 \\
uncorr-s-w\_10 & 338090 & 9.0 & c & 2509 & 17.1 \\
b-s-corr\_01 & 338090 & 0.0 & c & 851 & 92.6 \\
b-s-corr\_06 & 338090 & 0.0 & c & 971 & 55.4 \\
b-s-corr\_10 & 338090 & 0.0 & c & 1300 & 25.4 \\
\hline
\end{tabular}
&
\begin{tabular}{r||r|r|c||r|r}
\hline
instance & $m$ & $\alpha$ & {\textit{ver}} & $t^{100}$ & $\beta^{100}$\\
\hline 
\hline
\multicolumn{6}{c}{instance family \texttt{pla85900}} \\
\hline
uncorr\_01 & 85899 & 32.4 & c & 2582 & 72.8 \\
uncorr\_06 & 85899 & 13.5 & c & 3888 & 40.6 \\
uncorr\_10 & 85899 & 40.8 & u & 140 & 12.9 \\
uncorr-s-w\_01 & 85899 & 16.4 & c & 1707 & 89.2 \\
uncorr-s-w\_06 & 85899 & 12.3 & c & 2053 & 45.7 \\
uncorr-s-w\_10 & 85899 & 13.6 & u & 152 & 16.3 \\
b-s-corr\_01 & 85899 & 0.0 & c & 4021 & 92.6 \\
b-s-corr\_06 & 85899 & 0.0 & c & 1619 & 55.4 \\
b-s-corr\_10 & 85899 & 0.0 & c & 3550 & 25.4 \\
uncorr\_01 & 429495 & 32.5 & c & 3506 & 72.7 \\
uncorr\_06 & 429495 & 13.6 & c & 6416 & 40.5 \\
uncorr\_10 & 429495 & 40.4 & u & 538 & 13.0 \\
uncorr-s-w\_01 & 429495 & 16.3 & c & 2470 & 89.2 \\
uncorr-s-w\_06 & 429495 & 12.8 & c & 7918 & 46.3 \\
uncorr-s-w\_10 & 429495 & 13.2 & u & 585 & 16.5 \\
b-s-corr\_01 & 429495 & 0.0 & c & 3492 & 92.6 \\
b-s-corr\_06 & 429495 & 0.0 & c & 5835 & 55.2 \\
b-s-corr\_10 & 429495 & 0.0 & c & 6834 & 25.4 \\
uncorr\_01 & 858990 & 33.2 & c & 7213 & 71.6 \\
uncorr\_06 & 858990 & 13.6 & c & 5752 & 40.4 \\
uncorr\_10 & 858990 & 40.6 & u & 1895 & 13.1 \\
uncorr-s-w\_01 & 858990 & 16.4 & c & 5036 & 89.2 \\
uncorr-s-w\_06 & 858990 & 12.7 & c & 11793 & 46.3 \\
uncorr-s-w\_10 & 858990 & 13.2 & u & 15593 & 17.4 \\
b-s-corr\_01 & 858990 & 0.0 & c & 6066 & 92.6 \\
b-s-corr\_06 & 858990 & 0.0 & c & 14733 & 56.2 \\
b-s-corr\_07 & 858990 & 0.0 & c & 19346 & 26.4 \\
\hline
\end{tabular}
\end{tabular}
}
\end{table}

\section{Conclusion}

We have introduced a new non-linear knapsack problem where items during a travel along a fixed route have to be selected. We have shown that both the constrained and unconstrained version of the problem are $\mathcal{NP}$-hard. Our proposed pre-processing scheme can significantly decrease the size of instances making them easier for computation. The experimental results show that small sized instances can be solved to optimality in a reasonable time by the proposed exact approach. Larger instances can be efficiently handled by our approximate approach producing near-optimal solutions. 

As a future work, this problem has several natural generalizations. First, it makes sense to consider the case where the sequence of cities may be changed. This variant asks for the mutual solution of the traveling salesman and knapsack problems. Another interesting situation takes place when cities may be skipped because are of no worth, for example any item stored there imposes low or negative profit. Finally, the possibility to pickup and delivery the items is for certain one another challenging problem.

\section*{Acknowledgments}
This research was supported under the ARC Discovery Project DP130104395.

\bibliographystyle{splncs03}
\bibliography{references}{}

\end{document}